\documentclass[11pt]{article}
\usepackage[margin=1in]{geometry}
\usepackage[colorlinks=true, allcolors=blue]{hyperref}
\usepackage{CJKutf8}
\usepackage{fixmath}
\usepackage{bm}
\usepackage{amsbsy}
\usepackage{color}
\usepackage{verbatim}
\usepackage{multirow}
\usepackage{amssymb}
\usepackage{amsthm}
\usepackage{array}
\usepackage{mathtools}
\usepackage{amsmath}
\usepackage{graphicx}
\usepackage{tikz}
\usetikzlibrary{positioning}
\usepackage{xcolor}
\usepackage{thm-restate}

\newtheorem{theorem}{Theorem}
\newtheorem{lemma}[theorem]{Lemma}

\newtheorem{definition}[theorem]{Definition}

\newtheorem{claim}{Claim}

\newcommand{\size}[1]{\ensuremath{|#1|}}

\newcommand{\floor}[1]{\ensuremath{\lfloor#1\rfloor}}

\newcommand{\lrA}[1]{\ensuremath{\left(#1\right)}}

\def\OPT{\mbox{OPT}}

\def\P{\mathcal{P}}
\def\T{\mathcal{T}}

\title
{
The APX-hardness of the Traveling Tournament Problem
}

\author
{
Jingyang Zhao\\
University of Electronic Science and Technology of China\\
\texttt{jingyangzhao1020@gmail.com}
\and
Mingyu Xiao\\
University of Electronic Science and Technology of China\\
\texttt{myxiao@uestc.edu.cn}
}

\date{}

\begin{document}

\maketitle

\begin{abstract}
The Traveling Tournament Problem (TTP-$k$) is a well-known benchmark problem in sports scheduling, which asks us to design a double round-robin schedule such that each pair of teams plays one game in each other's home venue, no pair of teams plays each other on two consecutive days, each team plays at most $k$ consecutive home games or away games, and the total traveling distance of all the $n$ teams is minimized. TTP-$k$ allows a polynomial-time approximation scheme when $k=2$ and becomes APX-hard when $k\geq n-1$. In this paper, we reduce the gap by showing that TTP-$k$ is APX-hard for any fixed $k\geq3$.

\end{abstract}

\section{Introduction}
In the field of sports scheduling~\cite{kendall2010scheduling}, the traveling tournament problem (TTP-$k$) is a well-known benchmark problem that was first systematically introduced in~\cite{easton2001traveling}. This problem aims to find a double round-robin tournament satisfying some constraints, minimizing the total distance traveled by all participating teams. In a double round-robin tournament of $n$ teams, where $n$ is required to be even, each team will play two games against each of the other $n-1$ teams, one home game at its home venue and one away game at its opponent's home venue. Additionally, each team should play one game a day, all games need to be scheduled on $2(n-1)$ consecutive days, and so there are exactly $n/2$ games on each day. 
For TTP-$k$, we have the following three basic constraints or assumptions on the double round-robin tournament.

\medskip
\noindent
\textbf{Traveling Tournament Problem (TTP-$k$)}
\begin{itemize}
\item \emph{No-repeat}: Two teams cannot play against each other on two consecutive days.
\item \emph{Direct-traveling}: Before the first game, all teams are at home, and they will return home after the last game. Furthermore, a team travels directly from its game venue on the $i$th day to its game venue on the $(i+1)$th day.
\item \emph{At-most-$k$}: Each team can have at most $k$ consecutive home games or away games.
\end{itemize}

The input of TTP-$k$ is a complete graph where each vertex represents a team and the weight between two vertices $a$ and $b$, denoted by $w(a,b)$, is the distance from the home of team $a$ to the home of team $b$. We consider the case where $w$ is a (general) metric function, i.e., $w(a,b)+w(b,c)\geq w(a,c)=w(c,a)\geq 0$ and $w(a,a)=0$ for all vertices $a,b,c$. While this assumption is common, we provide one reason for adopting it: in the real-world, distance functions may be asymmetric or fail to satisfy the triangle inequality, and in such cases, using a metric (which allows for general forms of distance functions) can better model asymmetric and non-metric relationships and make the problem more realistic.

TTP-$k$ is a difficult optimization problem. Readers can refer to \cite{DBLP:journals/eor/BulckGSG20,duran2021sports} for an overview. When $k=1$, the problem is infeasible~\cite{de1988some}. When $k\geq2$, and $n\geq 2$ is an even number, feasible solutions exist. Next, we survey the complexity of TTP-$k$.

For $k=2$, NP-hardness is still unknown but a polynomial-time approximation scheme (PTAS) has already been found~\cite{thielen2012approximation,imahori20211+,DBLP:conf/ijcai/ZhaoX21}. For $k=3$, Thielen and Westphal~\cite{thielen2011complexity} proved NP-hardness by a reduction from 3-SAT. Using a generalized reduction from $k$-SAT, Chatterjee~\cite{DBLP:journals/corr/abs-2110-02300} proved NP-hardness for any fixed $k>3$. 
When $k\geq n-1$, the at-most-$k$ constraint loses its meaning and the traveling distance of a team could be as short as that in the Traveling Salesman Problem (TSP), and we use $k=n-1$ to denote this case.
For $k=n-1$, the problem is also known as the Unconstrained TTP. Bhattacharyya~\cite{bhattacharyya2016complexity} proved NP-hardness of TTP-$(n-1)$ using a reduction from $(1,2)$-TSP, i.e., a special case of TSP where the edge weights are either 1 or 2. Using a linear reduction (L-reduction) from $(1,2)$-TSP, Bendayan \emph{et al.}~\cite{2022apx} further proved APX-hardness of TTP-$(n-1)$. So, there is no PTAS for TTP-$(n-1)$ unless P=NP.

Since TTP-2 admits a PTAS while TTP-$(n-1)$ does not, a natural question is whether TTP-$k$ admits a PTAS or not for any fixed $k\geq 3$.
In this paper, we answer this question by proving that TTP-$k$ is APX-hard for any fixed $k\geq 3$. A summary of the current-known complexities of TTP-$k$ is shown in Table~\ref{overview}.

\begin{table}[ht]
\centering
\begin{tabular}{c|c|c|c|c}
\hline
    TTP-$k$ & $k=2$ & $k=3$ & $3<k=O(1)$ & $k\geq n-1$\\
\hline
    P v.s. NP-hard & Unknown & NP-hard~\cite{thielen2011complexity} & NP-hard~\cite{DBLP:journals/corr/abs-2110-02300} & NP-hard~\cite{bhattacharyya2016complexity}\\
\hline
    PTAS v.s. APX-hard & PTAS~\cite{thielen2012approximation,imahori20211+,DBLP:conf/ijcai/ZhaoX21} & \multicolumn{2}{c|}{APX-hard \textbf{[This Paper]}} & APX-hard~\cite{2022apx}\\
\hline
\end{tabular}
\caption{A summary of the current-known complexities of TTP-$k$.}
\label{overview}
\end{table}

It is worth mentioning that there is a large number of studies on approximation algorithms~\cite{thielen2012approximation,imahori20211+,DBLP:conf/ijcai/ZhaoX21,DBLP:conf/mfcs/XiaoK16,DBLP:conf/cocoon/ZhaoX21,DBLP:conf/atmos/ChatterjeeR21,miyashiro2012approximation,2012LDTTP,hoshino2013approximation}. For $k=2$, the approximation ratios for even $n/2$ and odd $n/2$ have been improved to $(1+3/n)$ and $(1+5/n)$~\cite{zhao2024practical}, which implies that TTP-2 admits a PTAS\footnote{Suppose we have a $(1+c/n)$-approximation algorithm for TTP-2 where $c$ is a constant, we can get a PTAS for TTP-2 as follows: for any constant $\varepsilon>0$, if $n\leq c/\varepsilon$, we can optimally solve TTP-2 in polynomial time by a brute-and-force algorithm; otherwise, we use the $(1+c/n)$-approximation algorithm, which admits an approximation ratio of $1+c/n\leq 1+\varepsilon$.}. For $k=3$ and $k=4$, the approximation ratios are $(139/87+\varepsilon)$~\cite{zhao2022improved} and $(13/8+\varepsilon)$~\cite{DBLP:conf/aaai/ZhaoX25}, respectively. For $5\leq k=o(n)$, the ratio is $(5k-7)/(2k)+O(k/n)$~\cite{yamaguchi2009improved}. For $k=n-1$, Imahori~\emph{et al.}~\cite{imahori2010approximation} proposed a $2.75$-approximation algorithm. At the same time, Westphal and Noparlik~\cite{westphal2014} proposed a $5.875$-approximation algorithm for all $k\geq 4$ and $n\geq 6$. This was later improved by Zhao and Xiao~\cite{DBLP:journals/corr/abs-2309-01902} to a $5$-approximation for any feasible choice of $k$ and $n$. Note that when $k$ is sufficiently large, the latter four approximation ratios can be slightly improved due to the $(3/2-10^{-36})$-approximation algorithm for TSP~\cite{DBLP:conf/stoc/KarlinKG21}.

\section{Preliminary}
We always consider $k$ as a constant. A simple path on $k$ different vertices is called a \emph{$k$-path}. A \emph{$k$-path packing} in a graph of $n$ vertices is a set of vertex-disjoint $n/k$ $k$-paths such that all vertices in the graph are covered. Note that a $k$-path packing exists only if $n\bmod k=0$. A weight function is called an \emph{integer} function if it only outputs non-negative and constant integers. If the input is a graph for an instance $I$, we use $I=(V_I,E_I)$ to denote the graph where $V_I$ denotes the vertex set and $E_I$ denotes the edge set of the graph. For any functions $f(n)$ and $g(n)$, we write $f(n)=\Omega(g(n))$ (resp., $f(n)=\Theta(g(n))$) to indicate that there exist constants $\alpha,n_0>0$ (resp., $n_0,\alpha',\alpha>0$) such that $f(n)\geq \alpha\cdot g(n)$ (resp., $\alpha'\cdot g(n)\leq f(n)\leq \alpha\cdot g(n)$) for all $n\geq n_0$.

\subsection{APX-hardness and L-reduction}
An optimization problem belongs to the APX class if it admits an $O(1)$-approximation algorithm. Moreover, it has a PTAS if it has a $(1+\varepsilon)$-approximation algorithm for any constant $\varepsilon>0$. A problem is APX-hard if for every problem in APX there is a PTAS-reduction to it. To prove APX-hardness of a problem, we may use an L-reduction from an APX-hard problem, where an L-reduction directly implies a PTAS reduction~\cite{crescenzi1997short}.
For any instance $I$ of any optimization problem, let $\OPT(I)$ denote the optimal value of the optimal solution of $I$. 
The definition of L-reduction is given as follows.

\begin{definition}[\cite{williamson2011design}]\label{L-reduction}
Given two optimization problems $\Pi$ and $\Xi$, we have an L-reduction from $\Pi$ to $\Xi$ if for some constants $\alpha, \beta > 0$:
\begin{enumerate}
    \item For each instance $I$ of $\Pi$ we can compute in polynomial time an instance $J$ of $\Xi$;
    \item $\OPT(J)\leq \alpha\cdot\OPT(I)$;
    \item Given a solution of weight $S(J)$ to $J$, we can compute in polynomial time a solution of weight $S(I)$ to $I$ such that
    \[
    \size{S(I)-\OPT(I)}\leq \beta\cdot\size{S(J)-\OPT(J)}.
    \]
\end{enumerate}
\end{definition}

\subsection{The Framework of the Reduction}
To establish APX-hardness of TTP-$k$, we employ a well-known APX-hard problem called the $k$-tour cover problem ($k$-TC). However, instead of directly using $k$-TC, we consider a variant called Boosted Restricted $k$-TC. We first show that Boosted Restricted $k$-TC remains APX-hard (see Section~\ref{introktc}). 

In our L-reduction from Boosted Restricted $k$-TC to TTP-$k$, given an instance $I$ of Boosted Restricted $k$-TC, we first construct an instance $I'$ of $k$-TC satisfying some desired properties and then construct an instance $J$ of TTP-$k$ in polynomial time. The details of the constructions are explained in Section~\ref{condition1}.
By further proving the last two properties of the L-reduction in Sections~\ref{condition2} and \ref{condition3} respectively, we prove APX-hardness of TTP-$k$.

\begin{theorem}\label{main}
TTP-$k$ is APX-hard for any fixed $k\geq 3$.
\end{theorem}

\section{The Boosted Restricted $k$-Tour Cover Problem}\label{introktc}
In TTP-$k$, the well-known Independent Lower Bound~\cite{easton2003solving,urrutia2007new} is computed by determining the minimum possible traveling distance for each team, independently of the feasibility of the other teams, and then summing these distances. Furthermore, to compute the minimum possible traveling distance for a single team, we need to solve $k$-TC~\cite{asano1997covering}, defined as follows. 

\begin{definition}[$k$-TC]\label{ktc}
An instance of $k$-TC can be presented by a complete graph $I=(V_I, E_I)$, where there is one depot vertex $o$ in $V_I$, $n-1$ non-depot vertices in $V_I\setminus\{o\}$, and a metric weight function $w$ on the edges in $E_I$. The goal of $k$-TC is to find a set of tours of minimum total weight covering all non-depot vertices where each tour is a simple cycle that contains the depot vertex and at most $k$ non-depot vertices.
\end{definition}

$k$-TC, also known as the (unit-demand) Capacitated Vehicle Routing Problem~\cite{dantzig1959truck}, is APX-hard for any fixed $k\geq 3$~\cite{asano1996covering,asano1996covering+}.


\begin{lemma}[\cite{asano1996covering,asano1996covering+}]\label{apxx}
$k$-TC is APX-hard for any fixed $k\geq 3$.
\end{lemma}

The technical report~\cite{asano1996covering} is unavailable; however, an alternative version exists with a Japanese title~\cite{asano1996covering+}. 
For the sake of completeness, we restate their proof and explain more of the details in Appendix~\ref{ADX}.

Given the relationship between TTP-$k$ and $k$-TC, one potential approach to prove APX-hardness of TTP-$k$ is through an L-reduction from $k$-TC to TTP-$k$. However, directly performing such a reduction is not a straightforward task. We first introduce a restricted version of $k$-TC, denoted as Restricted $k$-TC. It was implicitly used and proved to be APX-hard for any fixed $k\geq 3$ in~\cite{asano1996covering}.

\begin{definition}[Restricted $k$-TC]\label{restr}
Compared to the definition of $k$-TC, Restricted $k$-TC introduces an additional constraint on the weight function $w$. Specifically, the weight function $w$ is required to be an integer metric function, and it must satisfy the condition $w(o,v)=\Theta(1)$ for all $v\in V_I\setminus\{o\}$.
\end{definition}

Definition~\ref{restr} is motivated by the following fact.
The outputs of the weight function used in~\cite{asano1996covering} (see also Lemma~\ref{apxx}) are all constant fractions (related to $k$), e.g., $\frac{3k-7}{2k-4}$. Moreover, the weights of edges incident to the depot vertex are all positive. Therefore, by scaling up a positive constant integer (related to $k$), we can assume w.l.o.g.\ that the weight function is an integer function and it holds $w(o,v)=\Theta(1)$ for every $v\in V_I\setminus\{o\}$. 
For any solution of Restricted $k$-TC with weight $S(I)$, we can get $S(I)-\OPT(I)=\Omega(1)$ if $S(I)\neq \OPT(I)$. Moreover, the solution of Restricted $k$-TC has the following property.

\begin{lemma}\label{bounded}
For any instance $I$ of Restricted $k$-TC, any feasible solution has a weight of $\Theta(n)$.
\end{lemma}
\begin{proof}
Since each tour covers at most $k$ non-depot vertices, to cover all $n-1$ non-depot vertices, any feasible solution must contain at least $\floor{\frac{n-1}{k}}$ tours and thus use at least $2\cdot\floor{\frac{n-1}{k}}$ edges incident to the depot vertex. By Definition~\ref{restr}, $w(o,v)=\Theta(1)$ for any $v\in V_I\setminus\{o\}$. Hence, a lower bound of any feasible solution is $\Theta(n)$. Moreover, by the triangle inequality, any feasible solution has a weight of at most $2\sum_{v\in V_I}w(o,v)=\Theta(n)$, i.e., in the worst solution, each tour only covers a single non-depot vertex. Therefore, any feasible solution has a weight of $\Theta(n)$.
\end{proof}

Then, we introduce a boosted version of Restricted $k$-TC, denoted as Boosted Restricted $k$-TC. 

\begin{definition}[Boosted Restricted $k$-TC]\label{brestr}
Given an instance $I$ of Restricted $k$-TC, the objective of Boosted Restricted $k$-TC is to find a solution $I$ of Restricted $k$-TC with weight $S(I)$ that minimizes the expression $m(m^2-1)\cdot S(I)$, where $m=\Theta(n)$ (to be defined later).
\end{definition}

Boosted Restricted $k$-TC differs from Restricted $k$-TC only in the objective function, so any feasible solution to the former is a feasible solution to the latter, and vice versa. They also share the same input.
Note that our L-reduction is from Boosted Restricted $k$-TC to TTP-$k$, as this is easier to present than a direct L-reduction from Restricted $k$-TC to TTP-$k$. For convenience, we formally define this problem. 

The concept of the ``boosted method" has been used previously in proving APX-hardness of TTP-$(n-1)$ in~\cite{2022apx}. This method ensures that Boosted Restricted $k$-TC remains APX-hard for any fixed $k\geq3$. Otherwise, there exists a $(1+\varepsilon)$-approximation algorithm for Boosted Restricted $k$-TC, and it can generate a solution $S$ for Restricted $k$-TC with a weight of $S(I)$ satisfying $m(m^2-1)\cdot S(I)\leq (1+\varepsilon)\cdot m(m^2-1)\cdot OPT(I)$, i.e., $S(I)\leq (1+\varepsilon)\cdot OPT(I)$, which implies that the approximation algorithm is also a $(1+\varepsilon)$-approximation algorithm for Restricted $k$-TC, contradicting APX-hardness of Restricted $k$-TC for any fixed $k\geq 3$.

Next, we show the details of the L-reduction from Boosted Restricted $k$-TC to TTP-$k$. By Definition~\ref{L-reduction}, there are three conditions in the L-reduction. We will consider them respectively.

\section{The Construction of the Instance}\label{condition1}
Given an instance $I$ of Restricted $k$-TC, we make some modifications to obtain an instance $I'$ of $k$-TC and then show how to construct an instance $J$ of TTP-$k$, both in polynomial time.

An \emph{$l$-tour} is a simple cycle that contains the depot vertex and exactly $l$ non-depot vertices. A solution of an instance of $k$-TC that contains only $k$-tours is called a \emph{saturated} solution. An instance of $k$-TC is called a \emph{saturated} instance if it has an optimal saturated solution.
Not every instance $I$ of Restricted $k$-TC is saturated.
However, by adding some new non-depot vertices on the depot position, we can get an equivalent saturated instance $I'$ of $k$-TC, i.e., there is an optimal saturated solution with the same weight of the optimal solution in $I$. In the following, we use $m$ to denote the number of non-depot vertices in $I'$.

\begin{lemma}
Given an instance $I$ of Restricted $k$-TC, there is a polynomial-time algorithm to get an equivalent saturated instance $I'$ of $k$-TC with an even number of non-depot vertices.
\end{lemma}
\begin{proof}
%
Suppose that there is an optimal solution in $I$ that contains $c_l$ of $l$-tours, where $1\leq l\leq k$. We have $n-1=\sum_{l\leq k}c_ll$. Firstly, we can see that after adding some new non-depot vertices on the depot position we can get an equivalent new instance. If we add $\sum_{l<k}c_l(k-l)$ new non-depot vertices on the depot position, the optimal solution in $I$ corresponds to an optimal saturated solution in the new instance since we can fill each $l$-tour in the optimal solution in $I$ with $k-l$ of the new non-depot vertices. Then, if we further add $n'$ with $n'\bmod k=0$ new non-depot vertices on the depot position, there still exists an optimal saturated solution in the new instance since the new non-depot vertices can be covered by $n'/k$ new $k$-tours with a zero weight. 

We have $\sum_{l<k}c_l(k-l)<\sum_{l<k}nk<nk^2$ since $c_l<n$. So, let $n'=nk^2+k-((n-1)\bmod k)-\sum_{l<k}c_l(k-l)$, and then we have $n'>k-((n-1)\bmod k)>0$. Moreover, since $n-1=\sum_{l\leq k}c_ll$, we have $(n-1+\sum_{l<k}c_l(k-l))\bmod k=0$, and hence $n'\bmod k=0$. 

So, given the instance $I$, by adding $\sum_{l<k}c_l(k-l)+n'=nk^2+k-((n-1)\bmod k)$ new non-depot vertices on the depot position, we can get an equivalent saturated instance $I'$ of $k$-TC. Note that there is 1 depot vertex and $m\coloneqq n-1+nk^2+k-((n-1)\bmod k)$ non-depot vertices in $I'$. We assume w.l.o.g.\ that $m\bmod 2=0$. Otherwise, we have $k\bmod2=1$ since $m\bmod k=0$, and in this case we can further add $k$ new non-depot vertices on the depot position (these $k$ vertices can be covered by using a $k$-tour with a zero weight). Then, $m\leftarrow m+k$ would be even. Since $k=O(1)$, $I'$ can be obtained in polynomial time.
\end{proof}

Note that the original instance $I$ of Restricted $k$-TC has $n$ vertices, while the modified instance $I'$ of $k$-TC has $m+1$ vertices, including one depot vertex and 
\begin{equation}\label{m=tn}
m=\Theta(n) 
\end{equation} 
non-depot vertices. Specifically, if $n-1+nk^2+k-((n-1)\bmod k)$ is even, we have $m=n-1+nk^2+k-((n-1)\bmod k)$; otherwise, we have $m=n-1+nk^2+2k-((n-1)\bmod k)$.

To construct an instance $J$ of TTP-$k$, we further add $m(m^2-1)-1$ new vertices on the depot position of $I'$. So, $J$ has $m^3$ vertices, each representing a team. Among these vertices, $m(m^2-1)$ vertices are considered as \emph{dummy} vertices, including one depot vertex and $m(m^2-1)-1$ new vertices. This construction can be achieved in polynomial time. So, we have the following result.

\begin{lemma}\label{l1}
Given an instance $I$ of Restricted $k$-TC, we can construct the above instance $J$ of TTP-$k$ in polynomial time.
\end{lemma}

An illustration of the relations for instances $I$, $I'$ and $J$ can be seen in Figure~\ref{fig01}. Both $I'$ and $J$ are obtained by adding some vertices on the depot position of $I$.

\begin{figure}[ht]
\centering
\begin{tikzpicture}
\draw (0,0) circle (30pt);
\draw (1.43,0) ellipse [x radius=30pt, y radius=20pt];
\draw (0,0.53) circle (15pt);
\draw (0,-0.53) circle (15pt);

\node  at (-0.75,0) {C};
\node  at (0,0.9) {A};
\node  at (0,-0.9) {B};
\node  at (2,0) {$I$};

\draw
(0,0.2) circle (2pt)
(-0.2,0.6) circle (2pt) (0,0.6) circle (2pt) (0.2,0.6) circle (2pt)
(0,-0.2) circle (2pt)
(-0.2,-0.6) circle (2pt) (0,-0.6) circle (2pt) (0.2,-0.6) circle (2pt);
\node   (up)  at (0.04, 0.4) {$\dots$};
\node   (up)  at (0.04, -0.4) {$\dots$};
\node   (up)  at (1.45, 0.14) {$\vdots$};

\filldraw[black] 
(0.8,0) circle (2pt);
\draw 
(1.25,0) circle (2pt)
(1.65,-0.2) circle (2pt) (1.65,0) circle (2pt) (1.65,0.2) circle (2pt);
\end{tikzpicture}
\caption{An illustration of the relations for instances $I$, $I'$, and $J$, where cycle $C$ denotes the depot position, cycle $A$ denotes $m-(n-1)$ vertices on the depot, and cycle $B$ denotes $m(m^2-1)-1$ vertices on the depot: the ellipse can be seen as instance $I$ (containing 1 depot vertex, denoted by the black node, and $n-1$ non-depot vertices, denoted by the white nodes), the ellipse with  vertices in $A$ (resp., $A$ and $B$) can be seen as instance $I'$ (resp., $J$).}
\label{fig01}
\end{figure}

\section{The Second Property}\label{condition2}
In this section, we prove the second property of the L-reduction in Definition~\ref{L-reduction}. Note that $I$ is an instance of Restricted $k$-TC and our L-reduction is from Boosted Restricted $k$-TC. Hence, we need to prove that $\OPT(J)\leq \alpha\cdot m(m^2-1)\cdot\OPT(I)$, where $\alpha$ is a constant.

Recall that in $I'$ of $k$-TC there always exists an optimal saturated solution that contains only $k$-tours. So, such a solution corresponds to a $k$-path packing in the graph induced by the $m$ non-depot vertices, denoted by $\P^*_k$. Based on this $k$-path packing, we will construct a solution for instance $J$ of TTP-$k$, which is an upper bound of $OPT(J)$.

\subsection{The construction of the schedule}
There are $d\coloneqq m/k$ $k$-paths in $\P^*_k$. We will pack them as a \emph{super-team}, denoted by $U_1$. There are still $m(m^2-1)$ dummy vertices on the depot position, which will be divided into $m(m^2-1)/k$ $k$-paths with a zero weight and then further divided into $m^2-1$ dummy super-teams by arbitrarily packing $d$ $k$-paths as a super-team, denoted by $\{U_2,\dots,U_{m^2}\}$. Hence, there are $m^2$ super-teams in total, and in each super-team there are $d$ $k$-paths and hence $m$ vertices.

If a dummy team in $U_i$ ($i>1$) (i.e., a dummy vertex with its position on the depot position) always plays $k$ consecutive away games along each $k$-path in $U_1$ then its total traveling distance is exactly $\OPT(I')=\OPT(I)$. We will construct a feasible solution for TTP-$k$ such that all dummy teams play $k$ consecutive away games along each $k$-path in $U_1$. Our construction is based on the $k$-path packing construction in~\cite{pathpacking2023}. The main idea is to first arrange \emph{super-games} between super-teams and then extend the super-games into normal games.

The super-game schedule contains $m^2$ \emph{blocks}. The designs in the first and last blocks are slightly different from other blocks. We first consider the first $m^2-1$ blocks. In each of the first $m^2-1$ blocks there are $m^2/2$ super-games. In the first block, the $m^2/2$ super-games can be seen in Figure~\ref{fig02}, where each super-game is represented by a directed edge and all super-games are \emph{normal} super-games. The last super-team is denoted by a double-cycle node, while the others are denoted by single-cycle nodes.

\begin{figure}[ht]
\centering
\resizebox{0.6\textwidth}{!}{
\begin{tikzpicture}
[
leftsuperteam/.style={circle, draw=black!100, very thick, minimum size=12.8mm, inner sep=0pt},
normalsuperteam/.style={circle, draw=black!100, very thick, minimum size=12.8mm, inner sep=0pt},
rightsuperteam/.style={draw=black!100, very thick, minimum size=12.8mm, inner sep=0pt},
cyc/.style={circle, draw=black!100, very thick, minimum size=11mm, inner sep=0pt},
]
\node[leftsuperteam]        (u9)    at (-1.2,0) {$U_{m^2}$};
\node[cyc]        (u9+)   at (-1.2,0) {};
\node[normalsuperteam]      (u8)   at (1,0) {$U_{1}$};
\draw[very thick,->] (u9.east) to (u8.west);

\node[normalsuperteam]      (up)    at (1+1.5*1, 0.9555) {$U_2$};
\node[normalsuperteam]      (down)    at (1+1.5*1, -0.9555) {$U_{m^2-1}$};
\draw[very thick,<-] (down.north) to (up.south);

\node[normalsuperteam]      (up)    at (1+1.5*2, 0.9555) {$U_3$};
\node[normalsuperteam]      (down)    at (1+1.5*2, -0.9555) {$U_{m^2-1}$};
\draw[very thick,<-] (up.south) to (down.north);

\node[normalsuperteam]      (up)    at (1+1.5*3, 0.9555) {$U_{4}$};
\node[normalsuperteam]      (down)    at (1+1.5*3, -0.9555) {$U_{m^2-2}$};
\draw[very thick,->] (up.south) to (down.north);

\node      (up)    at (1+1.5*4, 0) {$\dots$};

\node[normalsuperteam]      (up)    at (1+1.5*5, 0.9555) {$U_{\frac{m^2}{2}}$};
\node[normalsuperteam]      (down)    at (1+1.5*5, -0.9555) {$U_{\frac{m^2}{2}+1}$};
\draw[very thick,->] (up.south) to (down.north);
\end{tikzpicture}
}
\caption{The super-game schedule in the first block.}
\label{fig02}
\end{figure}

In the second block, we change the positions of single-cycle nodes in the cycle $U_1U_2\dots U_{m^2-1}$ by moving one position in the clockwise direction, and we also change the direction of each edge except for the most left edge (i.e., the edge incident to the double-cycle node). The main difference with the first block is that the super-game for the most left edge is called the \emph{left} super-game, where we put a letter `L' to denote it. The other super-games are still normal super-games. So, there is $1$ left super-game and $m^2/2-1$ normal super-games in the second block. See Figure~\ref{fig03} for an illustration.

\begin{figure}[ht]
\centering
\resizebox{0.6\textwidth}{!}{\begin{tikzpicture}
[
leftsuperteam/.style={circle, draw=black!100, very thick, minimum size=12.8mm, inner sep=0pt},
normalsuperteam/.style={circle, draw=black!100, very thick, minimum size=12.8mm, inner sep=0pt},
rightsuperteam/.style={draw=black!100, very thick, minimum size=12.8mm, inner sep=0pt},
cyc/.style={circle, draw=black!100, very thick, minimum size=11mm, inner sep=0pt},
]
\node[leftsuperteam]        (u9)    at (-1.2,0) {$U_{m^2}$};
\node[cyc]        (u9+)   at (-1.2,0) {};
\node[normalsuperteam]      (u8)   at (1,0) {$U_{m^2-1}$};
\draw[very thick,->,above] (u9.east) to node {L} (u8.west);

\node[normalsuperteam]      (up)    at (1+1.5*1, 0.9555) {$U_1$};
\node[normalsuperteam]      (down)    at (1+1.5*1, -0.9555) {$U_{m^2-2}$};
\draw[very thick,->] (down.north) to (up.south);

\node[normalsuperteam]      (up)    at (1+1.5*2, 0.9555) {$U_2$};
\node[normalsuperteam]      (down)    at (1+1.5*2, -0.9555) {$U_{m^2-3}$};
\draw[very thick,->] (up.south) to (down.north);

\node[normalsuperteam]      (up)    at (1+1.5*3, 0.9555) {$U_3$};
\node[normalsuperteam]      (down)    at (1+1.5*3, -0.9555) {$U_{m^2-4}$};
\draw[very thick,->] (down.north) to (up.south);

\node      (up)    at (1+1.5*4, 0) {$\dots$};

\node[normalsuperteam]      (up)    at (1+1.5*5, 0.9555) {$U_{\frac{m^2}{2}-1}$};
\node[normalsuperteam]      (down)    at (1+1.5*5, -0.9555) {$U_{\frac{m^2}{2}}$};
\draw[very thick,<-] (up.south) to (down.north);
\end{tikzpicture}}
\caption{The super-game schedule in the second block.}
\label{fig03}
\end{figure}

In the third block, we change the positions of single-cycle nodes in the cycle $U_1U_2\dots U_{m^2-1}$ by moving one position in the clockwise direction, and change the direction of every edge (including the most left edge). So, there is still $1$ left super-game and $m^2/2-1$ normal super-games in the third block. See Figure~\ref{fig04} for an illustration. Moreover, the schedules for the other blocks (except for the last) are derived analogously. 

\begin{figure}[ht]
\centering
\resizebox{0.6\textwidth}{!}{\begin{tikzpicture}
[
leftsuperteam/.style={circle, draw=black!100, very thick, minimum size=12.8mm, inner sep=0pt},
normalsuperteam/.style={circle, draw=black!100, very thick, minimum size=12.8mm, inner sep=0pt},
rightsuperteam/.style={draw=black!100, very thick, minimum size=12.8mm, inner sep=0pt},
cyc/.style={circle, draw=black!100, very thick, minimum size=11mm, inner sep=0pt},
]
\node[leftsuperteam]        (u9)    at (-1.2,0) {$U_{m^2}$};
\node[cyc]        (u9+)   at (-1.2,0) {};
\node[normalsuperteam]      (u8)   at (1,0) {$U_{m^2-2}$};
\draw[very thick,<-,above] (u9.east) to node {L} (u8.west);

\node[normalsuperteam]      (up)    at (1+1.5*1, 0.9555) {$U_{m^2-1}$};
\node[normalsuperteam]      (down)    at (1+1.5*1, -0.9555) {$U_{m^2-3}$};
\draw[very thick,<-] (down.north) to (up.south);

\node[normalsuperteam]      (up)    at (1+1.5*2, 0.9555) {$U_1$};
\node[normalsuperteam]      (down)    at (1+1.5*2, -0.9555) {$U_{m^2-4}$};
\draw[very thick,<-] (up.south) to (down.north);

\node[normalsuperteam]      (up)    at (1+1.5*3, 0.9555) {$U_2$};
\node[normalsuperteam]      (down)    at (1+1.5*3, -0.9555) {$U_{m^2-5}$};
\draw[very thick,<-] (down.north) to (up.south);

\node      (up)    at (1+1.5*4, 0) {$\dots$};

\node[normalsuperteam]      (up)    at (1+1.5*5, 0.9555) {$U_{\frac{m^2}{2}-2}$};
\node[normalsuperteam]      (down)    at (1+1.5*5, -0.9555) {$U_{\frac{m^2}{2}-1}$};
\draw[very thick,->] (up.south) to (down.north);
\end{tikzpicture}}
\caption{The super-game schedule in the third block.}
\label{fig04}
\end{figure}

Before we introduce the super-games in the last block, we first explain how to extend the normal/left super-games in the first $m-1$ blocks to normal games.

\textbf{Normal super-games.}
Consider a normal super-game from super-teams $U_i$ to $U_j$. Recall that each super-team contains exactly $d$ $k$-paths. Assume that $U_i=\{\{x_0,\dots,x_{k-1}\},\dots,\{x_{kd-k},$ $\dots,x_{kd-1}\}\}$ and $U_j=\{\{y_0,\dots,y_{k-1}\},\dots,\{y_{kd-k},\dots,y_{kd-1}\}\}$, where $\{x_{ki'-k},\dots,x_{ki'-1}\}$ (resp., $\{y_{ki'-k},\dots,y_{ki'-1}\}$) denotes the $i'$th $k$-path in $U_i$ (resp., $U_j$). The normal super-game will be extended into normal games on $2kd$ days in the following way:
\begin{itemize}
    \item Team $x_{ki+i'}$ plays an away game with $y_{kj+j'}$ on $(2k(i+j)+i'+j')\bmod 2kd$ day,
    \item Team $x_{ki+i'}$ plays a home game with $y_{kj+j'}$ on $(2k(i+j)+i'+j'+k)\bmod 2kd$  day,
\end{itemize}
where $0\leq i,j\leq d-1$ and $0\leq i',j'\leq k-1$. An illustration of the normal games after extending one normal super-game for $k=3$ and $d=2$ is shown in Table~\ref{normal-example}.

\begin{table}[ht]
\centering
\begin{tabular}{m{0.2cm}<{\centering}|*{12}{m{0.2cm}<{\centering}}}
  & $0$ & $1$ & $2$ & $3$ & $4$ & $5$ & $6$ & $7$ & $8$ & $9$ & $10$ & $11$\\
\hline
  $x_0$ & $y_0$ & $y_1$ & $y_2$ & $\pmb{y_0}$ & $\pmb{y_1}$ & $\pmb{y_2}$ & $y_3$ & $y_4$ & $y_5$ & $\pmb{y_3}$ & $\pmb{y_4}$ & $\pmb{y_5}$\\
  $x_1$ & $\pmb{y_5}$ & $y_0$ & $y_1$ & $y_2$ & $\pmb{y_0}$ & $\pmb{y_1}$ & $\pmb{y_2}$ & $y_3$ & $y_4$ & $y_5$ & $\pmb{y_3}$ & $\pmb{y_4}$\\
  $x_2$ & $\pmb{y_4}$ & $\pmb{y_5}$ & $y_0$ & $y_1$ & $y_2$ & $\pmb{y_0}$ & $\pmb{y_1}$ & $\pmb{y_2}$ & $y_3$ & $y_4$ & $y_5$ & $\pmb{y_3}$\\
  $x_3$ & $y_3$ & $y_4$ & $y_5$ & $\pmb{y_3}$ & $\pmb{y_4}$ & $\pmb{y_5}$ & $y_0$ & $y_1$ & $y_2$ & $\pmb{y_0}$ & $\pmb{y_1}$ & $\pmb{y_2}$\\
  $x_4$ & $\pmb{y_2}$ & $y_3$ & $y_4$ & $y_5$ & $\pmb{y_3}$ & $\pmb{y_4}$ & $\pmb{y_5}$ & $y_0$ & $y_1$ & $y_2$ & $\pmb{y_0}$ & $\pmb{y_1}$\\
  $x_5$ & $\pmb{y_1}$ & $\pmb{y_2}$ & $y_3$ & $y_4$ & $y_5$ & $\pmb{y_3}$ & $\pmb{y_4}$ & $\pmb{y_5}$ & $y_0$ & $y_1$ & $y_2$ & $\pmb{y_0}$\\
\hline
  $y_0$ & $\pmb{x_0}$ & $\pmb{x_1}$ & $\pmb{x_2}$ & $x_0$ & $x_1$ & $x_2$ & $\pmb{x_3}$ & $\pmb{x_4}$ & $\pmb{x_5}$ & $x_3$ & $x_4$ & $x_5$\\
  $y_1$ & $x_5$ & $\pmb{x_0}$ & $\pmb{x_1}$ & $\pmb{x_2}$ & $x_0$ & $x_1$ & $x_2$ & $\pmb{x_3}$ & $\pmb{x_4}$ & $\pmb{x_5}$ & $x_3$ & $x_4$\\
  $y_2$ & $x_4$ & $x_5$ & $\pmb{x_0}$ & $\pmb{x_1}$ & $\pmb{x_2}$ & $x_0$ & $x_1$ & $x_2$ & $\pmb{x_3}$ & $\pmb{x_4}$ & $\pmb{x_5}$ & $x_3$\\
  $y_3$ & $\pmb{x_3}$ & $\pmb{x_4}$ & $\pmb{x_5}$ & $x_3$ & $x_4$ & $x_5$ & $\pmb{x_0}$ & $\pmb{x_1}$ & $\pmb{x_2}$ & $x_0$ & $x_1$ & $x_2$\\
  $y_4$ & $x_2$ & $\pmb{x_3}$ & $\pmb{x_4}$ & $\pmb{x_5}$ & $x_3$ & $x_4$ & $x_5$ & $\pmb{x_0}$ & $\pmb{x_1}$ & $\pmb{x_2}$ & $x_0$ & $x_1$\\
  $y_5$ & $x_1$ & $x_2$ & $\pmb{x_3}$ & $\pmb{x_4}$ & $\pmb{x_5}$ & $x_3$ & $x_4$ & $x_5$ & $\pmb{x_0}$ & $\pmb{x_1}$ & $\pmb{x_2}$ & $x_0$\\
\end{tabular}
\caption{
Extending the normal super-game from $U_i=$ $\{\{x_0,x_1,x_2\}, \{x_3,x_4,x_5\}\}$ to $U_j=\{\{y_0,y_1,y_2\}$, $\{y_3,y_4,y_5\}\}$ into normal games on $2kd=12$ days, where home normal games are marked in bold.}
\label{normal-example}
\end{table}

We can see that every normal team in $U_i$ will play $d$ away trips along the $d$ $k$-paths in $U_j$, and every normal team in $U_j$ will play $d$ or $d+1$ away trips along at least $d-1$ $k$-paths in $U_i$. Moreover, all games are arranged between one team in $U_i$ and one in $U_j$. Note that this construction was firstly used in~\cite{hoshino2013approximation}.

\textbf{Left super-games.}
Consider a left super-game from super-teams $U_i$ to $U_j$. Assume that $U_i=\{\{x_0,\dots,x_{k-1}\},\dots,\{x_{kd-k},$ $\dots,x_{kd-1}\}\}$ and $U_j=\{\{y_0,\dots,y_{k-1}\},\dots,\{y_{kd-k},\dots,y_{kd-1}\}\}$. The extension is easy to present. For all normal teams in $U_i$ and $U_j$, we define matches $s_i=\{x_{i'}\rightarrow y_{(kd+i-i')\bmod kd}\}_{i'=0}^{kd-1}$, and also let $\overline{s_i}$ denote the case that the game venues in $s_i$ are reversed. We directly extend the left super-game into normal games on $2kd$ days in the following way:
\[
s_0\overline{s_1}s_2\overline{s_3}\cdots s_{kd-2}\overline{s_{kd-1}} \cdot\overline{s_0\overline{s_1}s_2\overline{s_3}\cdots s_{kd-2}\overline{s_{kd-1}}}.
\]
An illustration of the normal games after extending one normal super-game for $k=3$ and $d=2$ is shown in Table~\ref{left-example}. All games between one team in $U_i$ and one in $U_j$ are arranged.

\begin{table}[ht]
\centering
\begin{tabular}{m{0.2cm}<{\centering}|*{12}{m{0.2cm}<{\centering}}}
& $0$ & $1$ & $2$ & $3$ & $4$ & $5$ & $6$ & $7$ & $8$ & $9$ & $10$ & $11$\\
\hline
  $x_0$ & $y_0$ & $\pmb{y_1}$ & $y_2$ & $\pmb{y_3}$ & $y_4$ & $\pmb{y_5}$ & $\pmb{y_0}$ & $y_1$ & $\pmb{y_2}$ & $y_3$ & $\pmb{y_4}$ & $y_5$\\
  $x_1$ & $y_5$ & $\pmb{y_0}$ & $y_1$ & $\pmb{y_2}$ & $y_3$ & $\pmb{y_4}$ & $\pmb{y_5}$ & $y_0$ & $\pmb{y_1}$ & $y_2$ & $\pmb{y_3}$ & $y_4$\\
  $x_2$ & $y_4$ & $\pmb{y_5}$ & $y_0$ & $\pmb{y_1}$ & $y_2$ & $\pmb{y_3}$ & $\pmb{y_4}$ & $y_5$ & $\pmb{y_0}$ & $y_1$ & $\pmb{y_2}$ & $y_3$\\
  $x_3$ & $y_3$ & $\pmb{y_4}$ & $y_5$ & $\pmb{y_0}$ & $y_1$ & $\pmb{y_2}$ & $\pmb{y_3}$ & $y_4$ & $\pmb{y_5}$ & $y_0$ & $\pmb{y_1}$ & $y_2$\\
  $x_4$ & $y_2$ & $\pmb{y_3}$ & $y_4$ & $\pmb{y_5}$ & $y_0$ & $\pmb{y_1}$ & $\pmb{y_2}$ & $y_3$ & $\pmb{y_4}$ & $y_5$ & $\pmb{y_0}$ & $y_1$\\
  $x_5$ & $y_1$ & $\pmb{y_2}$ & $y_3$ & $\pmb{y_4}$ & $y_5$ & $\pmb{y_0}$ & $\pmb{y_1}$ & $y_2$ & $\pmb{y_3}$ & $y_4$ & $\pmb{y_5}$ & $y_0$\\
\hline
  $y_0$ & $\pmb{x_0}$ & $x_1$ & $\pmb{x_2}$ & $x_3$ & $\pmb{x_4}$ & $x_5$ & $x_0$ & $\pmb{x_1}$ & $x_2$ & $\pmb{x_3}$ & $x_4$ & $\pmb{x_5}$\\
  $y_1$ & $\pmb{x_5}$ & $x_0$ & $\pmb{x_1}$ & $x_2$ & $\pmb{x_3}$ & $x_4$ & $x_5$ & $\pmb{x_0}$ & $x_1$ & $\pmb{x_2}$ & $x_3$ & $\pmb{x_4}$\\
  $y_2$ & $\pmb{x_4}$ & $x_5$ & $\pmb{x_0}$ & $x_1$ & $\pmb{x_2}$ & $x_3$ & $x_4$ & $\pmb{x_5}$ & $x_0$ & $\pmb{x_1}$ & $x_2$ & $\pmb{x_3}$\\
  $y_3$ & $\pmb{x_3}$ & $x_4$ & $\pmb{x_5}$ & $x_0$ & $\pmb{x_1}$ & $x_2$ & $x_3$ & $\pmb{x_4}$ & $x_5$ & $\pmb{x_0}$ & $x_1$ & $\pmb{x_2}$\\
  $y_4$ & $\pmb{x_2}$ & $x_3$ & $\pmb{x_4}$ & $x_5$ & $\pmb{x_0}$ & $x_1$ & $x_2$ & $\pmb{x_3}$ & $x_4$ & $\pmb{x_5}$ & $x_0$ & $\pmb{x_1}$\\
  $y_5$ & $\pmb{x_1}$ & $x_2$ & $\pmb{x_3}$ & $x_4$ & $\pmb{x_5}$ & $x_0$ & $x_1$ & $\pmb{x_2}$ & $x_3$ & $\pmb{x_4}$ & $x_5$ & $\pmb{x_0}$\\
\end{tabular}
\caption{
Extending the left super-game from $U_i=$ $\{\{x_0,x_1,x_2\}, \{x_3,x_4,x_5\}\}$ to $U_j=\{\{y_0,y_1,y_2\}$, $\{y_3,y_4,y_5\}\}$ into normal games on $2kd=12$ days, where home normal games are marked in bold.}
\label{left-example}
\end{table}

\begin{lemma}\label{feas1}
There are no more than $k$ consecutive home/away games in the first $m^2-1$ blocks.
\end{lemma}
\begin{proof}
We use `A' and `H' to denote an away game and a home game, respectively.

By the design of left and normal super-games, it is clear that games in each block satisfy the constraint. Assume that there are more than $k$ consecutive home/away games for team $t_i$ at the joint of blocks $Q$ and $Q+1$. Recall that the states of each team on days $2i-1$ and $2i$ in each left super-game are AH or HA. So, team $t_i$ cannot be from super-team $U_{m^2}$ since $U_{m^2}$ always plays left super-games. Assume that team $t_i$ is from super-team $U_{i'}$ ($i'\neq m$). Moreover, there should be a normal super-game for $U_{i'}$ in block $Q$ or $Q+1$. For the sake of presentation, we use $L$ and $N$ (resp., $\overline{L}$ and $\overline{N}$) to indicate an away (resp, a home) left and normal super-game of super-team $U_{i'}$. According to the rotation scheme of the construction, there are six cases for the super-games of super-team $U_{i'}$ in blocks $Q$ and $Q+1$: $NL$, $NN$, $\overline{L}N$, and $\overline{NL}$, $\overline{NN}$, $\overline{\overline{L}N}$. Note that the first three cases are symmetrical with the last three. It is sufficient to consider the first three cases. In an away normal super-game, team $t_i$ plays $k'$ consecutive homes games on the first $k'$ days and $(k-k')$ consecutive homes games on the last $k-k'$ days where $0\leq k'\leq k-1$ depending on the position of $t_i$ in the $k$-path. For the joint of $NL$, since the states of $t_i$ are AH on the first two days of $L$, there are no consecutive home/away games. For the joint of $NN$, there are exactly $k$ consecutive home games. For the joint of $\overline{L}N$, since the states of $t_i$ are AH on the last two days of $\overline{L}$, there are at most $k'+1\leq k$ consecutive home games. Therefore, two consecutive blocks cannot create more than $k$ consecutive home/away games.
\end{proof}

At last, we are ready to design games in the last block. Recall that in a normal/left super-game between super-teams $U_i$ and $U_j$, all games between one team in $U_i$ and one in $U_j$ are arranged. Moreover, it is well-known that the rotation scheme of the construction can make sure any pair of $m^2$ super-teams meets exactly once in the first $m^2-1$ blocks. Hence, the unarranged games are games within each super-team. The unarranged games within each super-team can be seen as a double round-robin on $m$ teams (recall that $m$ is even). To arrange them, we will simply call a TTP-2 algorithm. Clearly, there would be no more than $2$ consecutive home/away games in the last block. We also need to make sure that there are no more than $k$ consecutive home/away games between the last two blocks. Thus, we will use a special TTP-2 algorithm where every team plays AH or HA on the first two days of its schedule.

\begin{lemma}
There is a special TTP-2 algorithm where every team plays AH or HA on the first two days of its schedule.
\end{lemma}
\begin{proof}
Assume that there are $m$ teams, denoted by $\{1,2,..,m\}$. 

\begin{figure}[ht]
\tiny
\begin{minipage}[t]{0.33\columnwidth}
\centering
\begin{tikzpicture}
[
leftsuperteam/.style={circle, draw=black!100, very thick, minimum size=7.5mm, inner sep=0pt},
normalsuperteam/.style={circle, draw=black!100, very thick, minimum size=7.5mm, inner sep=0pt},
rightsuperteam/.style={draw=black!100, very thick, minimum size=7.5mm, inner sep=0pt},
cyc/.style={circle, draw=black!100, very thick, minimum size=9.5mm, inner sep=0pt},
]
\node[leftsuperteam]        (u9)    at (-0.1,0) {$m$};
\node[cyc]        (u9+)   at (-0.1,0) {};
\node[normalsuperteam]      (u8)   at (1,0) {$1$};
\draw[very thick,->] (u9+.east) to (u8.west);

\node[normalsuperteam]      (up)    at (2, 0.55) {$2$};
\node[normalsuperteam]      (down)    at (2, -0.55) {$m-1$};
\draw[very thick,<-] (down.north) to (up.south);

\node[normalsuperteam]      (up)    at (3, 0.55) {$3$};
\node[normalsuperteam]      (down)    at (3, -0.55) {$m-2$};
\draw[very thick,<-] (up.south) to (down.north);

\node[normalsuperteam]      (up)    at (4, 0.55) {$4$};
\node[normalsuperteam]      (down)    at (4, -0.55) {$m-3$};
\draw[very thick,->] (up.south) to (down.north);

\node      (up)    at (4.5, 0) {$\dots$};
\end{tikzpicture}
{\footnotesize\centerline{(a) The first day}}
\end{minipage}
\begin{minipage}[t]{0.33\columnwidth}
\centering
\begin{tikzpicture}
[
leftsuperteam/.style={circle, draw=black!100, very thick, minimum size=7.5mm, inner sep=0pt},
normalsuperteam/.style={circle, draw=black!100, very thick, minimum size=7.5mm, inner sep=0pt},
rightsuperteam/.style={draw=black!100, very thick, minimum size=7.5mm, inner sep=0pt},
cyc/.style={circle, draw=black!100, very thick, minimum size=9.5mm, inner sep=0pt},
]
\node[leftsuperteam]        (u9)    at (-0.1,0) {$m$};
\node[cyc]        (u9+)   at (-0.1,0) {};
\node[normalsuperteam]      (u8)   at (1,0) {$m-1$};
\draw[very thick,<-] (u9+.east) to (u8.west);

\node[normalsuperteam]      (up)    at (2, 0.55) {$1$};
\node[normalsuperteam]      (down)    at (2, -0.55) {$m-2$};
\draw[very thick,<-] (down.north) to (up.south);

\node[normalsuperteam]      (up)    at (3, 0.55) {$2$};
\node[normalsuperteam]      (down)    at (3, -0.55) {$m-3$};
\draw[very thick,<-] (up.south) to (down.north);

\node[normalsuperteam]      (up)    at (4, 0.55) {$3$};
\node[normalsuperteam]      (down)    at (4, -0.55) {$m-4$};
\draw[very thick,->] (up.south) to (down.north);

\node      (up)    at (4.5, 0) {$\dots$};
\end{tikzpicture}
{\footnotesize\centerline{(b) The second day}}
\end{minipage}
\begin{minipage}[t]{0.33\columnwidth}
\centering
\begin{tikzpicture}
[
leftsuperteam/.style={circle, draw=black!100, very thick, minimum size=7.5mm, inner sep=0pt},
normalsuperteam/.style={circle, draw=black!100, very thick, minimum size=7.5mm, inner sep=0pt},
rightsuperteam/.style={draw=black!100, very thick, minimum size=7.5mm, inner sep=0pt},
cyc/.style={circle, draw=black!100, very thick, minimum size=9.5mm, inner sep=0pt},
]
\node[leftsuperteam]        (u9)    at (-0.1,0) {$m$};
\node[cyc]        (u9+)   at (-0.1,0) {};
\node[normalsuperteam]      (u8)   at (1,0) {$m-2$};
\draw[very thick,->] (u9+.east) to (u8.west);

\node[normalsuperteam]      (up)    at (2, 0.55) {$m-1$};
\node[normalsuperteam]      (down)    at (2, -0.55) {$m-3$};
\draw[very thick,<-] (down.north) to (up.south);

\node[normalsuperteam]      (up)    at (3, 0.55) {$1$};
\node[normalsuperteam]      (down)    at (3, -0.55) {$m-4$};
\draw[very thick,<-] (up.south) to (down.north);

\node[normalsuperteam]      (up)    at (4, 0.55) {$2$};
\node[normalsuperteam]      (down)    at (4, -0.55) {$m-5$};
\draw[very thick,->] (up.south) to (down.north);

\node      (up)    at (4.5, 0) {$\dots$};
\end{tikzpicture}
{\footnotesize\centerline{(c) The third day}}
\end{minipage}
\caption{An illustration of the special TTP-2 algorithm, where we only show the first three days for $m$ teams.}
\label{fig05}
\end{figure}
The framework of the special TTP-2 algorithm is similar to the construction of super-teams. The last team is denoted by a double-cycle node, and the other teams are single-cycle nodes. On the first day (see (a) in Figure~\ref{fig05}), except for the most left edge, the directions of edges alternatively change from left to right. On the next day, we change the positions of white nodes in the cycle $123...m-1$ by moving one position in the clockwise direction, but we only change the direction of the most left edge. The third day is obtained by the same rule. Based on this, we can obtain the games on the first $m-1$ days. Let $\Gamma_i$ denote the games on $i$th day, and $\overline{\Gamma_i}$ denote the games on $i$th day but with venues reversed. The complete schedule is presented by
\[
\Gamma_1\cdot \Gamma_2\cdots \Gamma_{m-2}\cdot \Gamma_{m-1}\cdot \overline{\Gamma_{m-2}\cdot \Gamma_{m-1}\cdot \Gamma_1\cdot \Gamma_2\cdots \Gamma_{m-3}}.
\]

  
  
  
  
  

The construction can be seen as a simple variant of TTP-2 construction in~\cite{thielen2012approximation}. It is easy to verify that the obtained schedule is feasible for TTP-2~\cite{thielen2012approximation}. Moreover, every team plays AH or HA on the first two days (see Figure~\ref{fig05}). Specifically, half teams play AH and half teams play HA on the first two days.
\end{proof}

Then, we use the special TTP-2 algorithm to design the unarranged games for each super-team. Consider a super-team $U_i$. If $U_i$ plays a left super-game, we know that all teams play AH or HA on the last two days in the $(m^2-1)$th block. By calling the special TTP-2 algorithm, every team plays AH or HA on the first two days in the last block. So, these two blocks cannot create more than $k$ consecutive home/away games. If $U_i$ plays an away (resp., home) normal super-game, there are $m/k$ teams playing $k$ consecutive home (resp., away) games and $m-$ $m/k$ teams play $k'$ consecutive home (resp., away) games on the last $k'$ days in the $(m^2-1)$th block. To avoid creating more than $k$ consecutive home/away games, we carefully call the special TTP-2 algorithm such that the $m/k<m/2$ teams play AH (resp., HA) (and each of the other $m-m/k$ teams will play AH or HA) on the first two days in the last block. This is feasible since there are exactly a half number ($m/2$) of teams playing AH and a half number of teams playing HA on the first two days of the special TTP-2 double round-robin. So, there are no more than $k$ consecutive home/away games between the last two blocks. Then, by Lemma~\ref{feas1}, we have the following result.

\begin{lemma}\label{feas2}
There are no more than $k$ consecutive home/away games in the construction.
\end{lemma}

\begin{theorem}\label{feas}
Given instance $J$, the above construction generates a feasible solution for TTP-$k$.
\end{theorem}
\begin{proof}
First, the solution is a double round-robin since two games between a pair of teams are arranged entirely in either a super-game or the TTP-2 double round-robin. Then, by the design of super-games and the definition of TTP-2 double round-robin, the games are not arranged on two consecutive days. So, the solution satisfies the no-repeat constraint.
Last, by Lemma \ref{feas2}, the solution also satisfies the at-most-$k$ constraint.
\end{proof}

\subsection{The quality of the schedule}
To analyze the quality, we calculate the weight of the itinerary for each team, i.e., the traveling distance for each team.

\textbf{Case~1.} First, we consider teams in $U_i\in\{U_2,\dots,U_{m^2}\}$. In each of the first $m^2-1$ blocks, $U_i$ plays an away normal super-game with $U_1$. By the design of normal super-games, all teams in $U_i$ play $d$ $k$ consecutive away games along the $d$ $k$-paths in $U_1$. Moreover, all teams are at home before the first game starts and after the last game ends in the away normal super-game. Hence, for each team, the trips of games in the super-game form a sub-itinerary. Recall that teams in $U_i$ are all dummy teams with positions on the depot. So, each team has a sub-itinerary same as the optimal saturated solution in $I'$ and hence has a traveling distance of $\OPT(I')=\OPT(I)$. Note that the other sub-itineraries for visiting teams in $U_{i'}$ ($i'\neq 1$) have a zero weight since all teams in super-teams of $\{U_2,\dots,U_{m^2}\}$ are dummy teams on the same position. So, these $m(m^2-1)$ dummy teams have a total traveling distance of $m(m^2-1)\cdot \OPT(I)$.

\textbf{Case~2.} Then, we consider teams in $U_1$. For the sake of analysis, we assume that all teams in $U_1$ return home after the last game ends in the $(m^2-1)$th block. So, it will split the itinerary of each team into two parts without decreasing the weight by the triangle inequality. Next, we consider a team $x\in U_1$.

\textbf{Case~2.1.} We first consider the first part of its itinerary, i.e., the trips due to the games in the first $m^2-1$ blocks. Note that $U_1$ plays $m^2-1$ home normal super-games in the construction. By the design of normal super-games, it will first play $k'$ consecutive away games ($0\leq k'<k$), then repeat $k$ consecutive home games with $k$ consecutive away games, then play $k$ consecutive home games, and last play $(k-k')$ consecutive away games. Hence, combining all these $m^2-1$ home normal games, we know that $x$ plays exactly $d(m^2-1)$ trips if $k'=0$ and $d(m^2-1)+1$ trips otherwise. Again, the trips are all for dummy teams, and hence each trip has a weight of $2w(o,x)$. Since each team in $U_1$ corresponds to a vertex in $V_{I'}\setminus\{o\}$, teams in $U_1$ have a total traveling distance of at most $2(d(m^2-1)+1)\sum_{x\in V_{I'}}w(o,x)$. Recall that $I'$ is obtained by adding some vertices on the depot of $I$. So, we have $\sum_{x\in V_{I'}}w(o,x)=\sum_{x\in V_{I}}w(o,x)$. Then, we have
\begin{align*}
2(d(m^2-1)+1)\sum_{x\in V_{I'}}w(o,x)&=2(d(m^2-1)+1)\sum_{x\in V_{I}}w(o,x)\\
&=2d(m^2-1)\sum_{x\in V_{I}}w(o,x) + 2\sum_{x\in V_{I}}w(o,x).
\end{align*}

\textbf{Case~2.2.} We then consider the second part of its itinerary, i.e., the trips due to the games in the last block. Recall that in the last block we call a TTP-2 algorithm for teams in $U_1$. In the worst case, team $x$ visits other $m-1$ teams using $m-1$ trips by the triangle inequality, i.e., each trip visits only a single team. So, teams in $U_1$ have a total traveling distance of at most $2\sum_{x\in V_{I'}}\sum_{y\in V_{I'}}w(x,y)$. By the triangle inequality, we have
\begin{align*}
2\sum_{x\in V_{I'}}\sum_{y\in V_{I'}}w(x,y)&\leq 2\sum_{x\in V_{I'}}\sum_{y\in V_{I'}}(w(o,x)+w(o,y))\\
&=2\sum_{x\in V_{I'}}\sum_{y\in V_{I'}}w(o,x)+2\sum_{x\in V_{I'}}\sum_{y\in V_{I'}}w(o,y)\\
&=2(m+1)\sum_{x\in V_{I'}}w(o,x)+2(m+1)\sum_{y\in V_{I'}}w(o,y)\\
&=4(m+1)\sum_{x\in V_{I'}}w(o,x)=4(m+1)\sum_{x\in V_{I}}w(o,x),
\end{align*}
where the second equality follows from that there is 1 depot vertex and $m$ non-depot vertices in $I'$, i.e., $\sum_{x\in V_{I'}}1=m+1$, and the last follows from that $I'$ is obtained by adding some vertices on the depot of $I$. 

Therefore, the schedule has a weight of at most $m(m^2-1)\cdot \OPT(I)+2d(m^2-1)\sum_{x\in V_{I}}w(o,x)+(4m+6)\sum_{x\in V_{I}}w(o,x)$.

\begin{lemma}\label{l2}
There exists a constant $\alpha>0$ such that $\OPT(J)\leq \alpha\cdot m(m^2-1)\cdot \OPT(I)$.
\end{lemma}
\begin{proof}
Given an instance $I$ of Restricted $k$-TC, our construction shows that 
\begin{align*}
\OPT(J)&\leq m(m^2-1)\cdot \OPT(I)+2d(m^2-1)\sum_{x\in V_{I}}w(o,x)+(4m+6)\sum_{x\in V_{I}}w(o,x).    
\end{align*}
By Definition~\ref{restr} and Lemma~\ref{bounded}, we have 
\begin{equation}\label{m=tn+}
\sum_{x\in V_{I}}w(o,x)=\Theta(n)\quad \mbox{and}\quad \OPT(I)=\Theta(n).  
\end{equation} 
Moreover, since $d=m/k$ and $m=\Theta(n)$ by (\ref{m=tn}), we can get
\begin{align*}
&m(m^2-1)\cdot \OPT(I)+2d(m^2-1)\sum_{x\in V_{I}}w(o,x)+(4m+6)\sum_{x\in V_{I}}w(o,x)\\
&=\Theta(m(m^2-1)\cdot \OPT(I)).    
\end{align*}
So, there exists a constant $\alpha>0$ such that $\OPT(J)\leq \alpha\cdot m(m^2-1)\cdot \OPT(I)$.
\end{proof}

\section{The Third Property}\label{condition3}
In this section, we prove the last property of the L-reduction in Definition~\ref{L-reduction}.
Given a solution $S(J)$ of instance $J$, we first show how to get a solution $S(I)$ of $I$ in polynomial time. 

Denote the itinerary consisting of the away trips by team $x\in V_J$ in $S(J)$ by $T_x$ and the traveling distance by $w(T_x)$. Recall that there are $m(m^2-1)$ dummy teams with their positions on the depot, denoted by $D_J$. 

\begin{lemma}\label{l3.1}
Given a solution $S(J)$ of instance $J$, we can construct a solution $S(I)$ of instance $I$ in polynomial time. 
\end{lemma}
\begin{proof}
Consider a dummy team $x\in D_J$. By shortcutting all vertices on the depot except for itself, we obtain a solution of $I$, denoted by $S_x(I)$, since instance $J$ is obtained by adding some vertices on the depot of $I$ (see Figure~\ref{fig01}). Note that $S_x(I)=w(T_x)$ since $x$ is on the depot. We can get a solution $S(I)$ of instance $I$ by choosing the one from $\{S_x(I)\mid x\in D_J\}$ with the minimum weight, which clearly takes polynomial time.
\end{proof}

\begin{lemma}\label{l3.2}
There exists a constant $\beta>0$ such that 
\[
m(m^2-1)\cdot (S(I)-\OPT(I))\leq\beta\cdot(S(J)-\OPT(J)).
\]
\end{lemma}
\begin{proof}
By the proof of Lemma~\ref{l3.1}, we have
$
\sum_{x\in D_J}w(T_x)=\sum_{x\in D_J}S_x(I)\geq m(m^2-1)\cdot S(I).
$

Next, we consider a non-dummy team $x\in V_J\setminus D_J$. In its itinerary $T_x$, it uses at least $m(m^2-1)/k=d(m^2-1)$ trips to visit all dummy teams. So, we can get $w(T_x)\geq 2d(m^2-1)w(o,x)$ by the triangle inequality. Then, we have 
$
\sum_{x\in V_J\setminus D_J}w(T_x)\geq 2d(m^2-1)\sum_{x\in V_{I}}w(o,x).
$

Hence, we have 
\begin{align*}
S(J)&=\sum_{x\in D_J}w(T_x)+\sum_{x\in V_J\setminus D_J}w(T_x)\\
&\geq m(m^2-1)\cdot S(I)+2d(m^2-1)\sum_{x\in V_{I}}w(o,x).
\end{align*}

By the proof of Lemma~\ref{l2}, we have $\OPT(J)\leq m(m^2-1)\cdot \OPT(I)+2d(m^2-1)\sum_{x\in V_{I}}w(o,x)+(4m+6)\sum_{x\in V_{I}}w(o,x)$. Then,
\[
S(J)-\OPT(J)\geq m(m^2-1)\cdot (S(I)-\OPT(I))-(4m+6)\sum_{x\in V_{I}}w(o,x).
\]

We assume w.l.o.g.\ that $S(I)>\OPT(I)$; otherwise, the lemma holds trivially. By Definition~\ref{restr}, we have $S(I)-\OPT(I)=\Omega(1)$. So, $m(m^2-1)\cdot (S(I)-\OPT(I))=\Omega(n^3)$. By (\ref{m=tn}) and (\ref{m=tn+}), we have $(4m+6)\sum_{x\in V_{I}}w(o,x)=\Theta(n^2)$. Then, we have
\begin{align*}
&m(m^2-1)\cdot (S(I)-\OPT(I))-(4m+6)\sum_{x\in V_{I}}w(o,x)\\
&=\Theta(m(m^2-1)\cdot (S(I)-\OPT(I))).    
\end{align*}
So, there exists a constant $\beta>0$ such that $m(m^2-1)\cdot (S(I)-\OPT(I))\leq\beta\cdot(S(J)-\OPT(J))$.
\end{proof}

Theorem~\ref{main} directly follows from Lemmas \ref{l1}, \ref{l2}, \ref{l3.1} and \ref{l3.2}.

\section{Conclusion}
Previously, it was known that TTP-$k$ allows a PTAS when $k=2$ and is APX-hard when $k\geq n-1$ (i.e., the unconstrained TTP). We clarified the complexity status by showing that TTP-$k$ is APX-hard for any fixed $k\geq 3$. 
Our hardness result, instead of using $k$-SAT, is derived from a reduction from $k$-TC, which also provides an alternative proof for NP-hardness of TTP-$k$ with any fixed $k\geq 3$. 
We remark that our APX-hardness result may not extend to geometric settings, e.g., Euclidean distances. Whether a PTAS exists for TTP-$k$ under Euclidean distances is still unknown.

\begin{CJK}{UTF8}{min}
\bibliographystyle{plain}
\bibliography{main}
\end{CJK}

\newpage
\appendix

\section{Proof of Lemma~\ref{apxx}}\label{ADX}
\begingroup
\def\thelemma{\ref{apxx}}
\begin{lemma}[\cite{asano1996covering,asano1996covering+}]
$k$-TC is APX-hard for any fixed $k\geq 3$.
\end{lemma}
\begin{proof}
Our proof is based on the proof in~\cite{asano1996covering+}, but we provide more details, e.g., the proofs of Claims~\ref{c1} and \ref{c3}.

The reduction is from 
a well-known APX-hard problem~\cite{DBLP:journals/ipl/Kann94}, named Maximum $k$-Path Packing on Bounded-Degree Graphs, i.e., given a connected graph $G=(V,E)$ with maximum degree at most some constant related to $k$, find a set of vertex-disjoint $k$-paths in $G$ with the largest possible cardinality.
Assume that $V=\{z_1,...,z_n\}$. Given $G=(V,E)$, the instance of $k$-TC, denoted by $I=(V_I,E_I)$ is constructed as follows.

First, $V_I$ consists of a depot vertex $o$, $n$ vertices $v_1,...,v_n$, and $n$ \emph{multiple} vertices $u_1,...,u_n$, where each $u_i$ has a multiplicity of $k-1$, i.e., $k-1$ vertices with zero interdistance. Thus, $\size{V_I}=1+nk$. Second, for edges $ou_i, u_iv_i, v_iv_j$ with $i,j\in\{1,...,n\}$ and $i\neq j$, the weights are defined as follows: $w(o,u_i)=\frac{1}{4}$, $w(u_i,v_i)=\frac{3}{4}$, and $w(v_i,v_j)=\frac{3k-7}{2k-4}$ if $z_iz_j\in E$. Finally, for other edges, the weights are defined as the shortest path weights generated from the weights defined above. Specifically, since $\frac{3k-7}{2k-4}\geq 1$ when $k\geq 3$, we have $w(o,v_i)=1$ for each $i\in\{1,...,n\}$, $w(v_i,v_j)=2$ for each $z_iz_j\notin E$, and $w(u_i,u_j)=\frac{1}{2}$ and $w(u_i,v_j)=\frac{5}{4}$ for $i,j\in\{1,...,n\}$ and $i\neq j$.

A \emph{straight tour} is a tour from $o$ to some $v_i$ that picks the $k-1$ copies of $u_i$ on its way; a \emph{mini-tour} is a tour from $o$ to some $u_i$ that picks its $k-1$ copies only; and an \emph{$h$-chain tour}, where $h\leq k$, is a tour based on an $h$-path $z_{i_1}...z_{i_h}$ in $G$ and visits $v_{i_1},...,v_{i_h}$ in this order, without picking any $u_{i_j}$.

A solution to $I$ is called \emph{canonical} if it consists only of straight tours and $h$-chain tours with accompanying mini-tours, where $h\leq k$, and each $v_i$ or all copies of $u_i$ are contained in only one tour.

\begin{claim}\label{c1}
Any solution to $I$ can be transformed in polynomial time into a canonical solution without increasing the weight.    
\end{claim}
\begin{proof}
By shortcutting, we assume w.l.o.g.\ that the solution to $I$ consists of a set of simple tours $\mathcal{T}$. Moreover, each $v_i$ is contained in only one tour in $\mathcal{T}$ by simply skipping $v_i$ in all but one of the tours.
Let 
\[
E^*_I\coloneqq\{ou_i,u_iv_i\mid i\in\{1,...,n\}\}\cup \{v_iv_j\mid z_iz_j\in E\}.
\]

Suppose that there exists a simple tour $T\in\T$ that does not fall into any of the three defined types. 
If $T$ contains an edge $xy\notin E^*_I$ with $x\neq o$ and $y\neq o$, we replace $xy$ with $ox$ and $oy$. By definition, $w(x,y)=w(o,x)+w(o,y)$. Thus, in this case, $T$ can be replaced with two simple tours without increasing the weight. 
Thus, we assume that $T$ contains no such edges. 
Moreover, since $T$ is a simple tour, it must be in the form of $ou_io$ or $ou_iv_io$ or $ou_iv_i...v_jo$ or $ou_iv_i...v_ju_jo$ or $ov_io$ or $ov_i...v_jo$.
Note that each tour of the three types falls into one of the above forms.
Therefore, each tour in $\mathcal{T}$ falls into one of the above forms.

Initialize $\mathcal{T}^*\coloneqq \emptyset$. 
For each tour in $\mathcal{T}$,
\begin{itemize}
    \item if it is in the form of $ou_io$, we select a mini-tour $ou_io$ into $\mathcal{T}^*$; 
    \item if it is in the form of $ou_iv_io$ or $ov_io$, we select a straight tour $ou_iv_io$ into $\mathcal{T}^*$; 
    \item if it is in the form of $ou_iv_i...v_ju_jo$ or $ou_iv_i...v_jo$ or $ov_i...v_jo$, we select an $h$-chain tour $ov_i...v_jo$ into $\mathcal{T}^*$.
\end{itemize}
Clearly, $\mathcal{T}^*$ can be obtained in polynomial time and consists only of tours of the three types. Moreover, since $w(o,v_i)=w(o,u_i)+w(u_i,v_i)$ by definition, we have $w(\mathcal{T}^*)=w(\mathcal{T})$.

Note that each $v_i$ is still contained in only one tour in $\mathcal{T}$. Thus, each $u_i$ is contained in at least one tour of the form $ou_io$ or $ou_iv_io$ in $\mathcal{T}$ because a tour visiting more vertices cannot serve all $k-1$ copies of $u_i$ alone. Hence, all $k-1$ copies of $u_i$ are served by either a mini-tour or a straight tour in $\mathcal{T}^*$. 
Moreover, since $\mathcal{T}$ is a solution to $I$ and $\mathcal{T}^*$ is obtained by skipping or adding some $u_i$ in certain tours of $\mathcal{T}$, each $v_i$ is also contained in only one tour in $\mathcal{T}^*$.
Thus, $\mathcal{T}^*$ is a feasible solution. 

By shortcutting, we can ensure that $\mathcal{T}^*$ consists only of straight tours and $h$-chain tours with accompanying mini-tours, where $h\leq k$, and each $v_i$ or all copies of $u_i$ are contained in only one tour. Therefore, a canonical solution can be obtained in polynomial time with a weight of at most $w(\mathcal{T})$.
\end{proof}


\begin{claim}\label{c2}
Any canonical solution to $I$ can be transformed in polynomial time into a new canonical solution that consists only of straight tours and $k$-chain tours with accompanying mini-tours, without increasing the weight.
\end{claim}
\begin{proof}
We only need to compare an $h$-chain tour, together with the associated $h$ mini-tours, to the $h$ straight tours that cover the same set of vertices. The weight of the $h$-chain tour plus mini-tours is $2+(h-1)\frac{3k-7}{2k-4}+\frac{h}{2}$ while the weight of the straight tours is $2h$. 
Notably, we have 
\begin{align*}
&2+(h-1)\frac{3k-7}{2k-4}+\frac{h}{2}-2h=\frac{k-1-h}{2k-4}.
\end{align*}
Thus, we have $2+(h-1)\frac{3k-7}{2k-4}+\frac{h}{2}\geq 2h$ when $h\leq k-1$, and $2+(h-1)\frac{3k-7}{2k-4}+\frac{h}{2}< 2h$ when $h=k$. 

Therefore, any $h$-chain tour with $h\leq k-1$, together with the associated $h$ mini-tours can be replaced with the $h$ straight tours that cover the same set of vertices, without increasing the weight.
\end{proof}

Suppose that $OPT(G)=l$, i.e., the maximum number of vertex-disjoint $k$-paths in $G$ is $l$.
By Claims~\ref{c1} and \ref{c2}, there exists an optimal solution to $I$ that consists of $l$ $k$-chain tours and $lk$ mini-tours corresponding to the $l$ $k$-paths in $G$, as well as $n-lk$ straight tours, i.e., $OPT(I)=l(2+\frac{(3k-7)(k-1)}{2k-4})+\frac{lk}{2}+2(n-lk)$.

Suppose that $k$-TC admits a PTAS. By Claims~\ref{c1} and \ref{c2}, a $(1+\varepsilon)$-approximate solution to $I$ can be modified in polynomial time into a canonical solution corresponding to $l'$ $k$-paths in $G$ such that $l'(2+\frac{(3k-7)(k-1)}{2k-4})+\frac{l'k}{2}+2(n-l'k)\leq (1+\varepsilon)\cdot OPT(I)$, i.e.,
\begin{align*}
&l'\lrA{2+\frac{(3k-7)(k-1)}{2k-4}}+\frac{l'k}{2}+2(n-l'k)\\
&\leq (1+\varepsilon)\cdot\lrA{l\lrA{2+\frac{(3k-7)(k-1)}{2k-4}}+\frac{lk}{2}+2(n-lk)}.
\end{align*}
Since $2+\frac{(3k-7)(k-1)}{2k-4}+\frac{k}{2}-2k=\frac{-1}{2k-4}$, we have 
\begin{align*}
2n-\frac{l'}{2k-4}\leq (1+\varepsilon)\cdot\lrA{2n-\frac{l}{2k-4}}.
\end{align*}
Thus,
\begin{equation}\label{eeeeq}
l'\geq (1+\varepsilon)\cdot l-\varepsilon\cdot (4k-8)\cdot n>l-\varepsilon\cdot (4k-8)\cdot n.
\end{equation}

\begin{claim}\label{c3}
It holds that $l=\Omega(n)$.
\end{claim}
\begin{proof}
Since $G$ is a connected graph with bounded degree, say $\Delta$, we construct a rooted spanning tree $F$ of $G$, which contains at least $\log_\Delta n$ levels. Thus, if $n\geq \Delta^k$, there exists a $k$-path $v_1...v_k$ from $F$ such that $v_k$ is a leaf at the last level and $v_1$ lies $k-1$ levels above $v_k$. Let $F_v$ denote the sub-tree of $F$ rooted at $v$.

Given $F$, we greedily select a $k$-path $v_1...v_k$ from $F$ such that $v_k$ is a leaf at the last level and $v_1$ lies $k-1$ levels above $v_k$, and update $F$ by deleting the sub-tree $F_{v_1}$.
Note that $F_{v_1}$ contains at most $\Delta^k$ vertices. 
Thus, the number of deleted vertices is at most $\Delta^k$ every time. So, there are at least $\frac{n-\Delta^k}{\Delta^k}$ vertex-disjoint $k$-paths in $F$. Since $\Delta$ and $k$ are constants, we have $l=\Omega(n)$.
\end{proof}

By Claim~\ref{c3}, there exist constants $\gamma,n_0>0$ such that $n\leq \gamma\cdot l$ for all $n\geq n_0$. Thus, when $n\geq n_0$, by (\ref{eeeeq}), we have
\[
l'\geq l-\varepsilon\cdot (4k-8)\cdot \gamma\cdot l=(1-\varepsilon\cdot(4k-8)\cdot\gamma)\cdot l,
\] 
which implies a PTAS for the APX-hard problem. Thus, $k$-TC is APX-hard for any fixed $k\geq 3$.
\end{proof}

\end{document}